\documentclass[journal]{IEEEtran}

\IEEEoverridecommandlockouts                              
\usepackage{cite}

\usepackage[utf8]{inputenc}
\usepackage{textcomp}
\usepackage{tikz,pgfplots}
\usepackage{graphicx}
\usepackage{mathrsfs} 
\usepackage{amsmath,amssymb,amsfonts,graphicx,float}

\usepackage{amsthm}
\usepackage{mathtools}
\usepackage[version=4]{mhchem}
\usepackage{siunitx}
\usepackage{longtable,tabularx}
\usepackage{hyperref}
\usepackage{algorithm}
\usepackage{algpseudocode}

\usepackage{xcolor}

\newtheorem{assumption}{Assumption}
\newtheorem{theorem}{Theorem}

\newtheorem{remark}{Remark}

\title{Non-Parametric Neuro-Adaptive \textcolor{black}{Formation} Control}

\author{Christos K. Verginis, Zhe Xu, and Ufuk Topcu
\thanks{C. K. Verginis is with Uppsala University, Uppsala, Sweden. E-mail:christos.verginis@angstrom.uu.se. 
Z. Xu is with  Arizona State University, Tempe, Arizona, USA. E-mail: xzhe1@asu.edu.
 U. Topcu is with the University of Texas at Austin, Austin, Texas, USA. E-mail:utopcu@utexas.edu.
	}%
}

\begin{document}

	\maketitle


\begin{abstract}
We develop a learning-based algorithm for the distributed formation control
of networked multi-agent systems governed by unknown, nonlinear dynamics. 
Most existing algorithms either assume certain parametric forms for the unknown dynamic terms or resort to unnecessarily large control inputs in order to provide theoretical guarantees. 
The proposed algorithm avoids these drawbacks by integrating neural network-based learning with adaptive control in a two-step procedure. 
In the first step of the algorithm, each agent learns a controller, represented as a neural network, using training data that correspond to a collection of formation tasks and agent parameters. 
{These parameters and tasks are derived by varying the nominal agent parameters and a user-defined formation task to be achieved, respectively.
In the second step of the algorithm, each agent incorporates the trained neural network into an online and adaptive control policy in such a way that the behavior of the multi-agent closed-loop system satisfies the user-defined formation task.} Both the learning phase and the adaptive control policy are distributed, in the sense that each agent computes its own actions using only local information from its neighboring agents. 
The proposed algorithm does not use any a priori information on the agents' unknown dynamic terms or any approximation schemes. We provide formal theoretical guarantees on the achievement of the formation task.
\end{abstract}


\section{Introduction} \label{sec:intro}

During the last decades, decentralized control of networked multi-agent systems has attracted significant attention due to the great variety of its applications, including multi-robot systems, transportation, multi-point surveillance as well as biological systems \cite{jadbabaie2003coordination,olfati2007consensus,couzin2005effective}. 
{In such systems, each agent calculates its own actions based on local information, as modeled by a connectivity graph, without relying on any central control unit. 
This absence of central control and global information motivates leader-follower architectures, where 
a team of agents (followers) aims at following a pre-assigned leader agent  that holds information about the execution of a potential task.}
The coordination problem of leader–follower architectures has been the
focus of many works \cite{modares2017optimal,bechlioulis2016decentralized,verginis2019adaptive,ni2020adaptive,zhang2012adaptive,hu2014adaptive}  because of its numerous applications in various disciplines including autonomous vehicles coordination (satellite formation flying, cooperative search of unmanned aerial vehicles and synchronization of Euler–Lagrange systems), systems biology (control and synchronization in cellular networks), and power systems (control of renewable energy microgrids).

Although many works on distributed cooperative control consider known and simple dynamic models, there exist many practical engineering systems that cannot be modeled accurately and 
are affected by unknown exogenous disturbances. Thus, the design of control algorithms that are robust and adaptable to such uncertainties and disturbances is important. For multi-agent systems, ensuring robustness is particularly challenging due to the lack of global information and the interacting dynamics of the individual agents. {A promising step towards the control of systems with uncertain dynamics is the use of data obtained a priori from system runs. However, engineering systems often undergo purposeful modifications (e.g., substitution of a motor or link in a robotic arm or exposure to new working environments) or suffer gradual faults (e.g., mechanical degradation), which might change the systems' dynamics or operating conditions. Therefore, one cannot rely on the aforementioned data to provably guarantee the successful control of the system. 
On the other hand, the exact incorporation of these changes in the dynamic model, 
and consequently, the design of new model-based algorithms, can be a challenging and often impossible procedure. Hence, the goal in such cases is to exploit the data obtained a priori and construct intelligent online policies that achieve a user-defined task while
adapting to the aforementioned changes. 
 }

\subsection{Contributions}

{
This paper addresses the distributed coordination of networked multi-agent systems governed by unknown nonlinear dynamics. Our main contribution lies in the development of a distributed learning-based control algorithm
that provably \textit{guarantees} the accomplishment of a given multi-agent formation task without any a priori information on the underlying dynamics. The algorithm draws a novel connection between distributed learning with neural-network-based representations and adaptive feedback control, and consists of the following steps. 
Firstly, it trains a number of neural networks, one for each agent, to approximate controllers for the agents that accomplish the given formation task. The data used to train the neural networks consist of pairs of states and control actions of the agents that are gathered from runs of the multi-agent system. 
Secondly, it uses an online adaptive feedback control policy that guarantees accomplishment of the given formation task. 
Both steps can be executed in a distributed manner in a sense that each agent uses only local information, as modeled by a connectivity graph. 
Our approach builds on a combination of controllers trained off-line and on-line adaptations, which was recently shown to significantly enhance performance with respect to single use of the off-line part \cite{bertsekas2021lessons}.
Numerical experiments show the robustness and adaptability of the proposed algorithm to different formation tasks, interactions among the agents, and system dynamics. That is, the proposed algorithm is able to achieve the given formation task even when the neural networks are trained with data that correspond to different multi-agent dynamic models (resembling a change in the dynamics of the agents), as well as different formation tasks and interactions among the agents. 
This paper extends our preliminary version \cite{verginisAAMAS} by providing (1) formal guarantees on the theoretical correctness of the proposed algorithm, and (2) a larger variety of experimental results.  
}

\subsection{{Related Work}}

\noindent{
\textbf{{Robust and adaptive control:}}
A large class of works on multi-agent coordination with uncertain dynamics falls in the category of robust and adaptive control \cite{chen2016adaptive,wang2017distributed,liu2017adaptive,
rezaee2017adaptive,yang2018robust,
rahimi2019distributed,li2014distributed,
verginis2019adaptive,verginis2021adaptive,
wang2015robust, lu2016cooperative,munz2010robust}. 
Standard adaptive-control methodologies, however, assume certain linear parametric forms for the unknown terms of the dynamics, limiting the dynamic uncertainties to unknown \textit{constant} terms \cite{chen2016adaptive,wang2017distributed,liu2017adaptive,rezaee2017adaptive}. Additionally, many works that do not employ parametric assumptions consider dynamic uncertainties and disturbances that are uniformly bounded \cite{yang2018robust,rahimi2019distributed} or satisfy growth conditions \cite{li2014distributed,verginis2019adaptive,verginis2021adaptive}. The works \cite{wang2015robust,lu2016cooperative} use functions in the control design that are larger than the upper bounds of the unknown dynamic terms; such a condition requires some a priori information on these terms. The work \cite{munz2010robust} assumes that the unknown drift terms of the dynamics are passive, which is then exploited in the stability analysis.  }
{Multi-agent coordination with unknown nonlinear continuous dynamics has been also tackled in the literature by using the so-called funnel control, without using dynamic approximations \cite{bechlioulis2015robust,bechlioulis2016decentralized,verginis2017robust,verginis2019robust}. Nevertheless, funnel controllers depend on so-called reciprocal time-varying barrier functions that drive the control input unbounded
when the error approaches a pre-specified funnel, creating thus unnecessarily large control inputs that cannot be realized by the system's actuators. In this paper, we develop a distributed control algorithm that does not employ such reciprocal terms and whose correctness does not rely on any of the aforementioned assumptions. }

\noindent{
\textbf{{Learning-based control:}}
A large variety of works focuses on distributed learning-based control to achieve multi-agent coordination under uncertain dynamics \cite{liu2020neural,qin2019neural,peng2013distributed,
wen2016neural,cheng2010neural,yuan2018cooperative,
mei2014distributed}. Such works resort to neural-network approximations of the unknown dynamic terms. In particular, they assume that the unknown functions of the dynamics are approximated arbitrarily well as a single-layer neural network with \textit{known} radial-basis activation functions and a vector of unknown but \textit{constant} weights. However, the accuracy of such approximations depends on the size of that vector, i.e., the number of neural-network neurons, implying that an arbitrarily small approximation error might require arbitrarily many weights. Additionally, 
there are no guidelines for choosing the activation functions in practice. }
{
Multi-agent coordination with unknown dynamics has also been tackled via cooperative reinforcement learning with stochastic processes \cite{zhang2021multi,omidshafiei2017deep,
hernandez2019survey,foerster2017stabilising,
gupta2017cooperative,
zhang2020distributed,dall2013distributed,
kar2012qd,wang2002reinforcement,
wai2018multi,doan2019finite}. However, such works usually adopt the conservative assumption that the agents have access to the states and actions of all other agents in the learning, execution, or both phases \cite{foerster2017stabilising,gupta2017cooperative}. Moreover, these works exhibit scalability problems with respect to the number of agents \cite{hernandez2019survey}, or assume the availability of time or state discretizations of the underlying continuous-time and continuous-state models. Additionally, the related works on multi-agent cooperative reinforcement learning usually consider common or team-average reward functions for the agents \cite{zhang2021multi,dall2013distributed}, which cannot be easily extended to account for inter-agent formation specifications that we account for. When relative inter-agent formation specifications are considered, the environment becomes non-stationary creating problems in the theoretical convergence analysis \cite{zhang2021multi}.}

{
In this work, we develop a distributed neuro-adaptive control algorithm for the formation control of continuous-time and -state multi-agent systems with unknown nonlinear dynamics.
In contrast to the related works in the literature, we do not assume linear parametrizations \cite{chen2016adaptive,wang2017distributed}, neural-network approximations \cite{liu2020neural,qin2019neural}, global boundedness or growth conditions \cite{yang2018robust,li2014distributed,verginis2019adaptive}, passivity properties \cite{munz2010robust}, or known upper bounds \cite{wang2015robust,lu2016cooperative} for the unknown dynamic terms. 
According to the best of our knowledge, the distributed formation-control problem with unknown dynamics has not been solved in the absence of the aforementioned assumptions. 
}


The rest of the paper is organized as follows. Section \ref{sec:PF} describes the considered problem. We provide our theoretical results in Section \ref{sec:main res}, and Section \ref{sec:exps} verifies the proposed methodology through experimental evaluation. Finally, Section \ref{sec:concl} concludes the paper.

\section{Problem Formulation} \label{sec:PF}

Consider a networked multi-agent group comprised of a leader, indexed by $i=0$, and $N$ followers, with $\mathcal{N}\coloneqq\{1,\dots,N\}$. The leading agent acts as an exosystem that generates a desired command/reference trajectory for the multi-agent group. The followers, which have to be controlled,
evolve according to the second-order dynamics
\begin{subequations} \label{eq:dynamics}	
\begin{align}
	\dot{x}_{i,1} &= x_{i,2} \\
	\dot{x}_{i,2} &= f_i(x_i,t) + g_i({x}_i,t)u_i
\end{align}
\end{subequations}
where ${x}_i \coloneqq [x_{i,1}^\top, x_{i,2}^\top]^\top \in \mathbb{R}^{n}\times \mathbb{R}^n$ is the $i$th agent's state, assumed available for measurement by agent $i$, $f_i:\mathbb{R}^{2n}\times[0,\infty) \to \mathbb{R}^n$, $g_i:\mathbb{R}^{2n}\times[0,\infty) \to \mathbb{R}^n$ are unknown functions modeling the agent's dynamics, and $u_i$ is the $i$th agent's control input. The vector fields $f_i(\cdot)$ and $g_i(\cdot)$ are assumed to be locally Lipschitz in ${x}_i$ over $\mathbb{R}^{2n}$ for each fixed $t\geq 0$, and uniformly bounded in $t$ over $[t_0,\infty)$ for each fixed ${x}_i\in\mathbb{R}^{2n}$, for all $i\in\mathcal{N}$. 
{In contrast to the works of the related literature, we do not assume any knowledge of the structure, Lipschitz constants, or bounds of $f_i(\cdot)$ and $g_i(\cdot)$, and we do not use any scheme to approximate them.  The lack of such assumptions renders the multi-agent coordination problem significantly difficult, since there is no apparent way to counteract the effect of the unknown drift terms $f_i()$. Moreover, in contrast to the funnel-based schemes, we do not resort to the use of reciprocal-like terms to dominate $f_i()$. }
Nevertheless, we do require the following assumption on the control directions $g_i(\cdot)$: 
\begin{assumption} \label{ass:g pd}
The matrices $g_i({x}_i,t)$ are positive definite, for all ${x}_i \in \Omega_i$, $t \geq 0$, where $\Omega_i\subset \mathbb{R}^{2n}$ are compact sets, $i\in\mathcal{N}$.
\end{assumption}
Assumption \ref{ass:g pd} is a sufficiently controllability condition for \eqref{eq:dynamics} and is adopted in numerous related works (e.g., \cite{wen2016neural,bechlioulis2016decentralized,verginis2017robust,huang2006nonlinear}).
The dynamics \eqref{eq:dynamics}, subject to Assumption \ref{ass:g pd}, comprise a large class of nonlinear dynamical systems that capture contemporary engineering problems in mechanical, electromechanical and power electronics applications, such as rigid/flexible robots, induction motors and DC-to-DC converters, to name a few.
Systems not covered by (\ref{eq:dynamics}) or Assumption \ref{ass:g pd} consist of underactuated or non-holonomic systems, such as unicycle robots, underactuated aerial or underwater vehicles. Such systems require special attention and their study consist part of our future work.  
Finally, the second-order model \eqref{eq:dynamics} can be easily extended to account for higher-order integrator systems \cite{slotine1991applied}.

We use an undirected graph $\mathcal{G} \coloneqq (\mathcal{N},\mathcal{E})$ to model the communication among the agents, with $\mathcal{N}$ being the index set of the agents, and $\mathcal{E} \subseteq \mathcal{N}\times\mathcal{N}$ being the respective edge set, with $(i,i)\notin \mathcal{E}$ (i.e., simple graph). The adjacency matrix associated with the graph $\mathcal{G}$ is denoted by $\mathcal{A} \coloneqq [a_{ij}]\in\mathbb{R}^{N\times N}$, with $a_{ij} \in \{0,1\}$, $i,j\in\{1,\dots,N\}$. If $a_{ij} = 1$, then agent $i$ obtains information regarding the state ${x}_j$ of agent $j$ (i.e., $(i,j)\in\mathcal{E}$), whereas if $a_{ij} = 0$ then there is no state-information flow from agent $j$ to agent $i$ (i.e., $(i,j)\notin\mathcal{E}$). Furthermore, the set of neighbors of agent $i$ is denoted by $\mathcal{N}_i \coloneqq \{j\in\mathcal{N}:(i,j)\in\mathcal{E}\}$, and the degree matrix is defined as $\mathcal{D} \coloneqq \textup{diag}\{|\mathcal{N}_1|,\dots,|\mathcal{N}_N|\}$. Since the graph is undirected, the adjacency is a mutual relation, i.e., $a_{ij} = a_{ji}$, rendering $\mathcal{A}$ symmetric. The \textit{Laplacian} matrix of the graph is defined as $\mathcal{L} \coloneqq \mathcal{D} - \mathcal{A}$ and is also symmetric. The graph is \textit{connected} if there exists a path between any two agents. For a connected graph, it holds that $\mathcal{L}\bar{1} = 0$, where $\bar{1}$ is the vector of ones of appropriate dimension.

Regarding the leader agent,  we denote its state variables by $x_0 \coloneqq [x_{0,1}^\top, x_{0,2}]^\top$ $\in\mathbb{R}^{2n}$, and consider the second-order dynamics
\begin{align*}
	\dot{x}_{0,1}(t) &= x_{0,2}(t) \\
	\dot{x}_{0,2}(t) &= u_{0}(t)
\end{align*} 
where $u_0:[0,\infty) \to \mathbb{R}^n$ is a bounded command signal. However, the leader provides its state only to a subgroup of the $N$ agents. In particular, we model the access of the follower agents to the leader's state via a diagonal matrix $\mathcal{B} \coloneqq \textup{diag}\{b_1,\dots,b_N\} \in \mathbb{R}^{N\times N}$; if $b_i = 1$, then the $i$th agent has access to the leader's state, whereas it does not if $b_i = 0$, for $i\in\mathcal{N}$. Thus, we  also define the augmented graph as $\bar{\mathcal{G}} \coloneqq (\mathcal{N}\cup\{0\}, \bar{\mathcal{E}})$, where $\bar{\mathcal{E}} \coloneqq \mathcal{E} \cup \{ (0,i) : b_i = 1 \}$.  We further define 
$$ H\coloneqq (\mathcal{L} + \mathcal{B})\otimes I_n,$$
where $\otimes$ denotes the Kronecker product, 
as well as the stacked vector terms 
\begin{align*}
	{x}_1 &\coloneqq [x_{1,1}^\top,\dots,x_{N,1}^\top]^\top \in \mathbb{R}^{Nn} \\
	 {x}_2 &\coloneqq [x_{1,2}^\top,\dots,x_{N,2}^\top]^\top \in \mathbb{R}^{Nn} \\
	 {x} &\coloneqq [{x}_1^\top,\dots,{x}_N^\top]^\top \in \mathbb{R}^{2Nn} \\
	 \bar{x}_{0,1} &\coloneqq [x_{0,1}^\top,\dots,x_{0,1}^\top]^\top \in \mathbb{R}^{Nn} \\
	 \bar{x}_{0,2} &\coloneqq [x_{0,2}^\top,\dots,x_{0,2}^\top]^\top \in \mathbb{R}^{Nn} \\
	 \bar{x}_0 &\coloneqq [\bar{x}_{0,1}^\top,\bar{x}_{0,2}^\top]^\top \in \mathbb{R}^{2Nn}. 
\end{align*}
By further defining 
\begin{align*}
	{f}({x},t) &\coloneqq [f_1(x_1,t)^\top,\dots,f_N(x_N,t)^\top]^\top\in\mathbb{R}^{Nn}\\ {g}({x},t) &\coloneqq \textup{diag}\{g_1(x_1,t),\dots, g_N(x_N,t)\} \in\mathbb{R}^{Nn\times Nn}, \\
	{u} &\coloneqq [u_1^\top, \dots, u_N^\top]^\top \in \mathbb{R}^{Nn},
\end{align*}
the dynamics \eqref{eq:dynamics} can be written as 
\begin{subequations} \label{eq:dynamics stack}
\begin{align} 
	\dot{{x}}_1 &= {x}_2 \\
	\dot{{x}}_2 &= {f}({x},t) + {g}({x}(t),t){u}.
\end{align}
\end{subequations}
The goal of this work is to design a distributed control algorithm, where each agent has access only to its neighbors' information, to achieve a pre-specified geometric formation of the agents in $\mathbb{R}^n$. More specifically, consider for each agent $i\in\mathcal{N}$ the constants $c_{ij}$, $j\in \{0\} \cup \mathcal{N}_i$ prescribing a desired offset that agent $i$ desires to achieve with respect to the leader ($j=0$), and its neighbors ($j\in\mathcal{N}_i$). That is, each agent $i\in\mathcal{N}_i$ aims at achieving $x_{i,1} = x_{j,1} - c_{ij}$, for all $j\in\mathcal{N}_i$, and if $b_i=1$ (i.e., the agent obtains information from the leader), $x_{i,1} = x_{0,1} - c_{i0}$. 
Note that, in the case of undirected graph, $c_{ij} = -c_{ji}$, for all $(i,j)\in \mathcal{E}$, and we assume that the set 
\begin{align*}
	\{ {x}_1\in\mathbb{R}^{Nn} :& x_{i,1} - x_{j,1} + c_{ij} = 0, \forall (i,j)\in\mathcal{E}, \\
	& b_i(x_{i,1} - x_{0,1} + c_{i0}) = 0, \forall i\in\mathcal{N} \}
\end{align*} is non-empty in order for the formation specification to be feasible. 

Furthermore, we impose the following assumption on the graph connectivity:
\begin{assumption} \label{ass:graph connect}
	The graph $\mathcal{G}$ is connected and there exists at least one $i\in\mathcal{N}$ such that $b_i =1$.
\end{assumption}
The aforementioned assumption dictates that
$\mathcal{L} + \mathcal{B}$ is an irreducibly diagonally dominant M-matrix [15].
An M-matrix is a square matrix having its off-diagonal entries non-positive and all principal minors nonnegative, thus $\mathcal{L} + \mathcal{B}$ is positive definite [15]. 

We define now the error variables for each agent as
\begin{align} \label{eq:e error}
	e_{i,1} \coloneqq \sum_{j\in\mathcal{N}_i} (x_{i,1} - x_{j,1} + c_{ij}) + b_i(x_{i,1} - x_{0,1} + c_{i0})
\end{align}
for $i\in\mathcal{N}$, and the respective stack vector $${e}_1 \coloneqq [e_{1,1}^\top,\dots,e_{N,1}^\top]^\top.$$ 
Next, by employing the multi-agent graph properties, noticing that $(\mathcal{L} \otimes I_n) \bar{x}_{0,1} = 0$ and that $(\mathcal{L} + \mathcal{B})$ is positive definite, \eqref{eq:e error} can be written as 
\begin{align} \label{eq:e error stack}
	{e}_1 \coloneqq H({x}_1 - \bar{x}_{0,1} + \bar{c}),
\end{align}
where 
\begin{align} \label{eq:c bar}
	{{c} \coloneqq \begin{bmatrix}
		c_1 \\ \vdots \\
		c_N
	\end{bmatrix} \coloneqq
H^{-1} 
	\begin{bmatrix}
		\sum_{j\in\mathcal{N}_1} c_{1j} + b_1c_{10} \\
		\vdots \\
		\sum_{j\in\mathcal{N}_N} c_{Nj} + b_Nc_{N0} 
	\end{bmatrix}}	
\end{align}
stacks the relative desired offsets $c_i$ of the $i$th agent
with respect to the leader, as dictated by the desired formation
 specification.  In this way, the desired formation is expressed with
 respect to the leader state, and is thus achieved when the state
 $x_{i,1}$ of each agent approaches the leader state $x_{0,1}$ with the
 corresponding offset $c_i$, $i \in\mathcal{N}$. Therefore, the formation control problem is solved if the control algorithm drives the disagreement vector
\begin{align}
	{\delta}_1 \coloneqq \begin{bmatrix}
		\delta_{1,1} \\ \vdots \\ \delta_{N,1}
	\end{bmatrix} 
	\coloneqq {x}_1 - \bar{x}_{0,1} + {c} 
\end{align}
to zero. However,
the disagreement formation variables $\delta_{i,1}$, are
global quantities and thus cannot be measured distributively by each agent
based on the local measurements, as they involve information directly from the leader as well as from the
whole graph topology via employing the inverse of $\mathcal{L} + \mathcal{B}$
in \eqref{eq:c bar}. Nevertheless, from \eqref{eq:e error stack}, one obtains
\begin{align} \label{eq:delta epsilon ineq}
	\left\|{\delta}_1 \right\| \leq \frac{\left\|{e}_1\right\|}{\sigma_{\min}(H)}
\end{align}
where $\sigma_{\min}(\cdot)$ denotes the minimum singular value. Therefore, convergence of ${e}_1$ to zero, which we aim to guarantee, implies convergence of ${\delta}_1$ to zero. {We further define the augmented errors for each agent 
\begin{align}\label{eq:e_2}
	e_{i,2} \coloneqq \dot{e}_{i,1} + k_{i,1}e_{i,1}
\end{align}
where $k_{i,1}$ are positive constants, for all $i\in\mathcal{N}$, the respective stacked vector 
\begin{align*}
 	e_2 \coloneqq [e_{i,2}^\top,\dots,e_{N,2}^\top]^\top \in \mathbb{R}^{nN}
\end{align*}
and the total error vector $e \coloneqq [e_1^\top,e_2^\top]^\top$.
By using \eqref{eq:e error stack}, the total error dynamics can be written as 
\begin{subequations} \label{eq:e dynamics}
\begin{align} 
 \dot{e}_1 &= -K_1 e_1 + e_2 \\
 \dot{e}_2 &= H( f(x(e),t) + g(x(e),t)u - \ddot{\bar{x}}_{0,1}) - K_1^2e_1 + K_1e_2,
\end{align}
\end{subequations}
where $K_1 \coloneqq \textup{diag}\{k_{i,1},\dots,k_{N,1}\}  \otimes I_n$ and with a slight abuse of notation, we express $x$ as a function of $e$ through \eqref{eq:e error stack}.
}

Before proceeding, we define the tuple 
\begin{align} \label{eq:F instance}
	\mathcal{F} \coloneqq (x_0(t),{f},{g},{c},\bar{\mathcal{G}},{x}(0))
\end{align} 
as the ``formation instance", characterized by the leader profile, the agent dynamics, the desired formation offsets, the graph topology, and the initial conditions of the agents.

\begin{figure*}
	\centering
	\includegraphics[width=.9\textwidth]{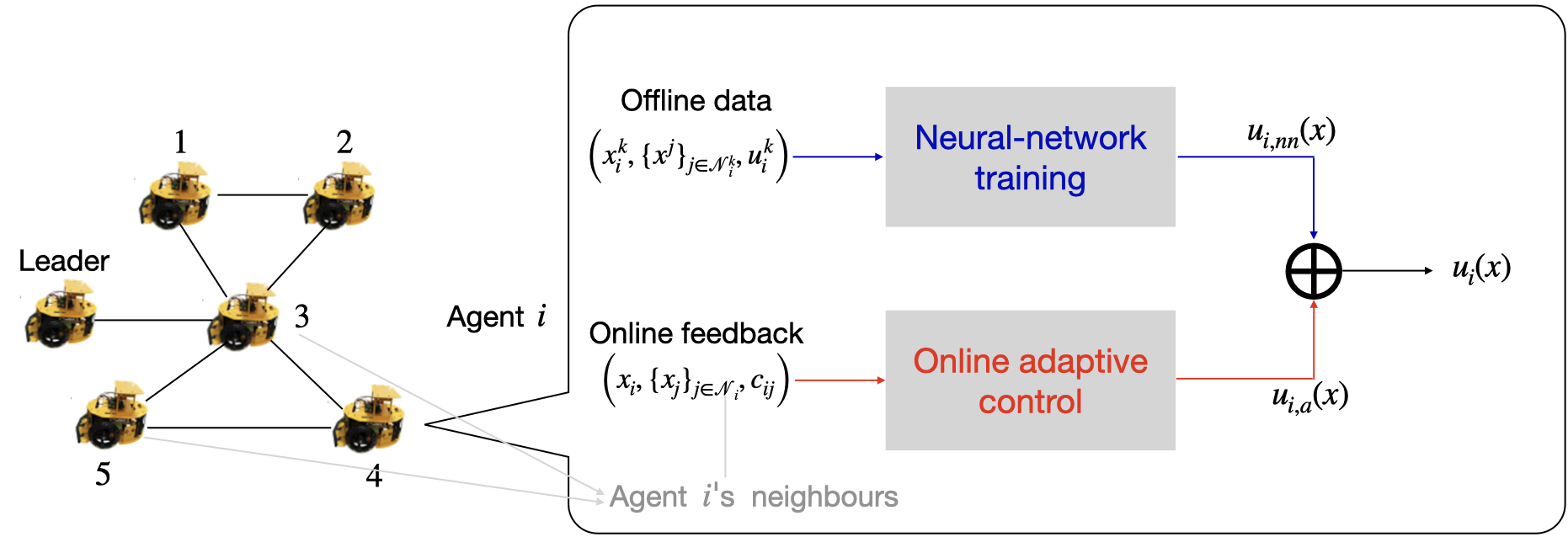}
	\caption{Diagram of the proposed two-step control algorithm. In the first, step, each agent learns a control policy $u_{i,\textup{nn}}(x)$ by training a neural network on the offline data $(x_i^k, \{x^j\}_{j\in \mathcal{N}_i^k}, u_i^k)$. In the second step, each agent uses an online feedback control policy that combines the trained neural network and an adaptive controller. }
	\label{fig:diagram}
\end{figure*}

\section{Main Results} \label{sec:main res}
This section describes the proposed algorithm, which consists of two steps. The first step consists of off-line learning of
distributed controllers, represented as neural networks, using training data derived from runs of the multi-agent system.
In the second step, we design an adaptive feedback
control policy that uses the neural networks and provably guarantees achievement of the formation specification.
{The proposed algorithm is depicted in Fig. \ref{fig:diagram}.}

\subsection{Neural-network learning} \label{subsec:NN}

As discussed in Section \ref{sec:intro}, we are inspired by cases where systems undergo changes that modify their dynamics and hence the underlying controllers no longer guarantee the satisfaction of a specific task. In such cases, instead of carrying out the challenging and tedious procedure of identification of the new dynamic models and design of new model-based controllers, we aim to exploit data from off-line system trajectories and develop a distributed online policy that is able to adapt to the aforementioned changes and achieve the formation task expressed via the offsets $c_{ij}$, $(i,j)\in\bar{\mathcal{E}}$. Consequently,
we assume the existence of data gathered from a finite set of $T$ trajectories $\mathcal{J}$ generated by a priori runs of the multi-agent system. 
More specifically, we consider that $\mathcal{J}$ is decomposed as $\mathcal{J} = (\mathcal{J}_1,\dots,\mathcal{J}_N)$, where $\mathcal{J}_i$ is the set of trajectories of agent $i\in\mathcal{N}$. 
Since the proposed control scheme is distributed, we consider that each agent $i$ has access to the data from its own set of trajectories $\mathcal{J}_i$, which comprises  
the finite set $$\mathcal{J}_i = \left\{{x}^k_i(t), \{{x}^j\}_{j\in\mathcal{N}^k_i},u^k_i \left({x}^k_i(t), \{{x}^j\}_{j\in\mathcal{N}^k_i},t \right) \right\}_{t\in \mathbb{T}_i}$$
where $\mathbb{T}_i$ is a finite set of time instants, ${x}^k_i\in\mathbb{R}^{2n}$ is the state trajectory of agent $i$ for trajectory $k$, $\mathcal{N}^k_i$ are the neighbors of agent $i$ in trajectory $k$, with $\{{x}^j\}_{j\in\mathcal{N}^k_i}$ being their respective state trajectories (which agent $i$ has access to, being their neighbor), and $u^k_i({x}^k_i(t), \{{x}^j\}_{j\in\mathcal{N}^k_i},t) \in \mathbb{R}^n$ is the  control input trajectory of agent $i$, which is a function of time and of its own and its neighbors' states. 

Each agent $i\in\mathcal{N}$ uses the data to train a neural network in order to approximate a controller that accomplishes the formation task. 
More specifically, each agent uses the tuples  $\{{x}^k_i(t), \{{x}^j\}_{j\in\mathcal{N}^k_i}\}_{t\in\mathbb{T}_i}$ as input to a neural network, and $u^k_i \big({x}^k_i(t), \{{x}^j\}_{j\in\mathcal{N}^k_i},t \big)_{t\in\mathbb{T}_i}$ as the respective output targets, for all $T$ trajectories. 
For the inputs corresponding to agents that are not neighbors of agent $i$ in a trajectory $k$, we disable the respective neurons. 
For a given ${x} \in \mathbb{R}^{2Nn}$, we denote by $u_{i,nn}({x})$ the output of the neural network of agent $ i \in \mathcal{N}$, and ${u}_{nn}({x}) \coloneqq [u_{1,nn}({x})^\top,\dots,u_{N,nn}({x})^\top]^\top$.

We stress that we do not require the training trajectories $\mathcal{J}$ to correspond to the formation instance $\mathcal{F}$ specified in \eqref{eq:F instance}. That is, each trajectory $k$ might be 
derived from the execution of a formation instance $\mathcal{F}_k = (x_0^k,{f}^k,{g}^k,{c}^k,\bar{\mathcal{G}}^k,{x}^k(0))$ that is different than the one specified in \eqref{eq:F instance}, i.e., different leader profile $x_0^k$, agent dynamics ${f}^k$, ${g}^k$, formation offsets ${c}^k$, communication graph $\bar{\mathcal{G}}^k$, and initial agent conditions ${x}^k(0)$, for all $k\in \mathbb{K}$ and some index set $\mathbb{K} \subset \mathbb{N}$.

Since the training trajectories are produced by the instances $\mathcal{F}_k$, which are different from $\mathcal{F}$, we do not expect the neural networks to learn how to achieve the formation task at hand, but rather to be able to adapt to the entire collection of tasks.  The motivation for training the neural networks with different tasks
and dynamics is the following. Since the tasks correspond to bounded trajectories, the respective stabilizing controllers compensate successfully the dynamics in \eqref{eq:dynamics}. Therefore,
the neural networks aim to approximate “average” distributed 
controllers that retain this property, i.e., the boundedness of the multi-agent dynamics \eqref{eq:dynamics}. By using such
approximation, the online feedback-control policy, which is illustrated in the next section, is able to 
guarantee achievement of the formation task at hand, without using any explicit information on the dynamics. 
We explicitly model the aforementioned approximation via the following assumption on the closed-loop system trajectory that is driven by the neural networks' output. 

{
\begin{assumption} \label{ass:nn}
	There exists $r > 0$ such that the stacked vector of outputs $u_{nn}(x)$ of the trained neural networks satisfies 
	\begin{align} \label{eq:nn cond}
		&\hspace{-3mm}  e_2^\top ( f(x(e),t) + g(x(e),t)u_{nn}(x) - \ddot{\bar{x}}_{0,1} )  \leq \kappa \|e_2\|^2
	\end{align}
for all $e$ satisfying $\|e\| \leq r$, where $\kappa$ is a positive constant.  
\end{assumption}
}

{Assumption \ref{ass:nn} is a sufficient condition for the prevention of finite-time escape of the  error trajectory $e(t)$ when the agents apply only the neural-network controllers, i.e., of the solution of the differential equation $\ddot{e}_2 = H( f(x,t) + g(x,t)u_{nn}(x) - \ddot{\bar{x}}_{0,1}) - K_1^2e_1 + K_1e_2$. 
Indeed, when the multi-agent system is driven solely by the neural-network controllers and satisfies \eqref{eq:nn cond}, one can find a Lyapunov function $V(e) = e^\top G e$, for a suitable constant matrix $G\in\mathbb{R}^{2nN\times 2nN}$, satisfying\footnote{$\lambda_{\min}$ and $\lambda_{\max}$ denote the minimum and maximum eigenvalues, respectively.} $\lambda_{\min}(G)\|e\|^2 \leq V(e) \leq \lambda_{\max}(G)\|e\|^2$ and $\dot{V} \leq \alpha \|e\|^2 \leq \frac{\alpha}{\lambda_{\min}\{G\}} V$, for all $\|e\| \leq r$ and a positive constant $\alpha$. Therefore, we conclude that $\lambda_{\min}(G)\|e(t)\|^2 \leq V(e(t)) \leq V(e(0)) \exp\left( \frac{\alpha}{\lambda_{\min}\{G\}} t \right)$, which prevents any finite-time escape of $e(t)$. Further note that the constants $r$ and $\kappa$ in \eqref{eq:nn cond} are \textit{unknown}. }

{
Assumption \ref{ass:nn} is motivated by
(i) the property of neural networks to approximate a continuous function arbitrarily well in a compact
domain for a large enough number of neurons and layers \cite{cybenko1989approximation}, and (ii) the fact that the neural networks are trained with bounded trajectories.  
As mentioned before, the collection of tasks that the neural networks
are trained with correspond to bounded trajectories. Hence, in view of the similarity of the dynamic terms that produce the training trajectories, the neural networks are expected to approximate a control policy that maintains the boundedness of the state trajectories as per \eqref{eq:nn cond}. 
Contrary to the related works (e.g., \cite{vamvoudakis2011multi,lewis2013cooperative,fei2020neural,liu2015neural,huang2006nonlinear,modares2017optimal}), however, we do not adopt approximation schemes for the system
dynamics. In fact, a standard assumption in the related literature is the approximation of an unknown function by a single-layer neural network as $\Theta(x) \vartheta + \epsilon$, 
 where $\Theta(x)$ is a \textit{known} matrix of radial basis function, $\vartheta$ is a vector of unknown constants, and $\epsilon$ is a constant error assumed sufficiently small. 
Nevertheless, Assumption \ref{ass:nn} is a less strict assumption; it does not require sufficiently good neural-network approximation through a sufficiently small error $\epsilon$ or 
knowledge of any radial-basis term $\Theta(x)$.  
Moreover, Assumption \ref{ass:nn} does not
imply that the neural-network outputs $u_{i,nn}({x},t)$ guarantee accomplishment of the formation task. It is merely a growth condition on the the solution of the system driven by $u_{nn}(x)$. In practice, \eqref{eq:nn cond} can be achieved by rich exploration of the state space by the leader agent $x_0^k$ in the training data $\mathcal{F}^k$. In the numerical experiments of Section \ref{sec:exps}, we show that \eqref{eq:nn cond}  holds true along the executed trajectories of the multi-agent system.}

{
We note that the neural-network controllers $u_{nn}$ can be replaced by other learning methodologies, as long as Assumption 3 holds. 
Nevertheless, the rich structure of neural networks makes them great candidates for approximating a control policy that satisfies \eqref{eq:nn cond}. 
}


\subsection{Distributed Control Policy}

We now design a distributed, adaptive feedback control policy to accomplish the formation task dictated by the graph topology $\bar{\mathcal{G}}$, the leader profile $x_0(t)$, and offsets $c_{ij}$, $(i,j)\in\bar{\mathcal{E}}$, given in Section \ref{sec:PF}. 



We define the adaptation variables $\hat{d}_{i,1}$ for each agent $i\in\mathcal{N}$, with $\hat{d}_1 \coloneqq [\hat{d}_{1,1},\dots,\hat{d}_{N,1}]^\top \in\mathbb{R}^N$, 
and design the distributed control policy as 
{
\begin{subequations} \label{eq:control and adapt}
\begin{align} \label{eq:control law}
	u_i = u_{i,nn}({x}) -(k_{i,2}  + \hat{d}_{i,1})e_{i,2} 
\end{align}
where $k_{i,2}$ are positive constants,
for all $i\in\mathcal{N}$. 
We further design the updates of the adaptation variables $\hat{d}_{i,1}$ as 
\begin{align}  \label{eq:adaptation law 1}
		\dot{\hat{d}}_{i,1} &\coloneqq \mu_{i,1}\|e_{i,2}\|^2 
\end{align}
\end{subequations}
where $\hat{d}_{i,1}(0) > 0$ and $\mu_{i,1}$ are  positive constants, for all $i\in\mathcal{N}$. 
}
\begin{figure*}
	\centering
	\includegraphics[width=.9\textwidth]{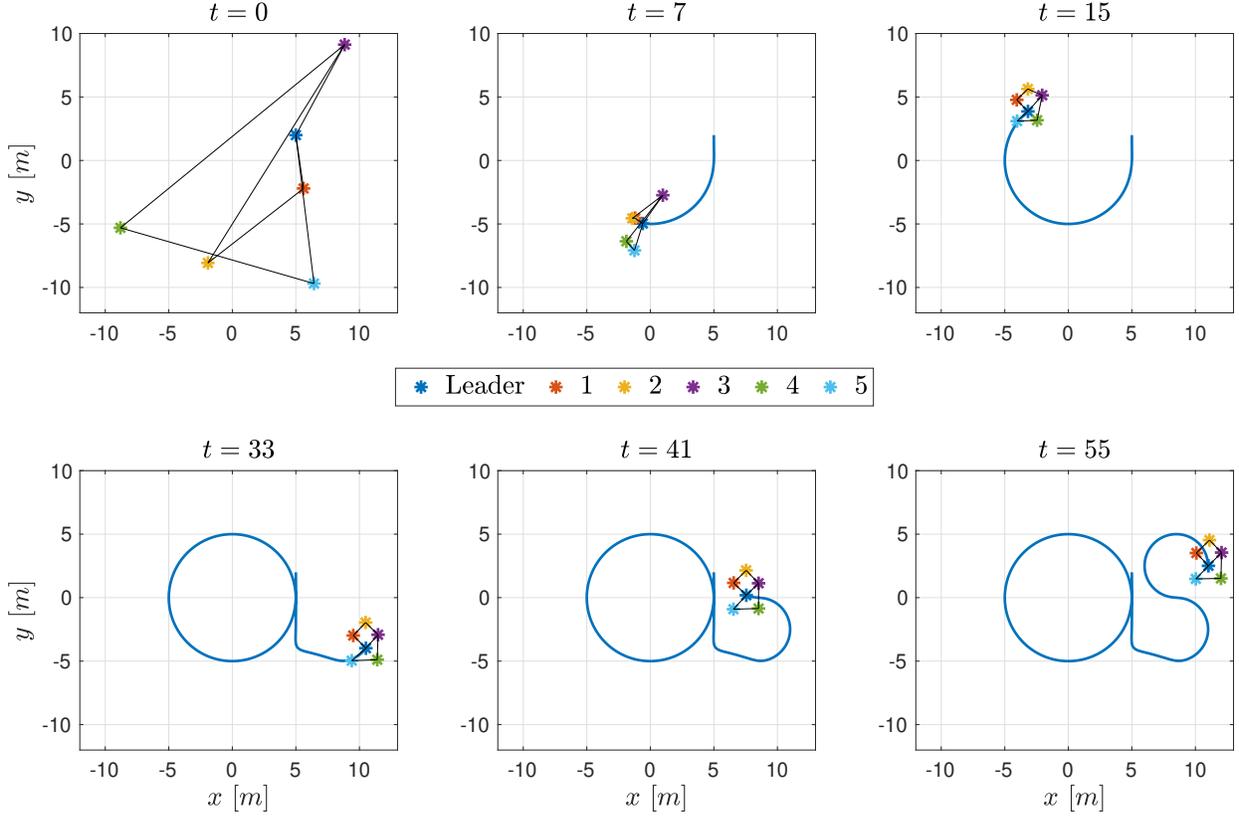}
	\caption{Snapshots of the first experiment in the $x$-$y$ plane. The agents converge to the desired formation (see bottom-middle and bottom-right plots) around the leader, which follows a pre-specified trajectory (continuous blue line). The black lines represent the communication edge set $\bar{\mathcal{E}}$ of the agents.}
	\label{fig:AS xy plot}
\end{figure*}

\begin{figure}
	\centering
	\includegraphics[width=.4\textwidth]{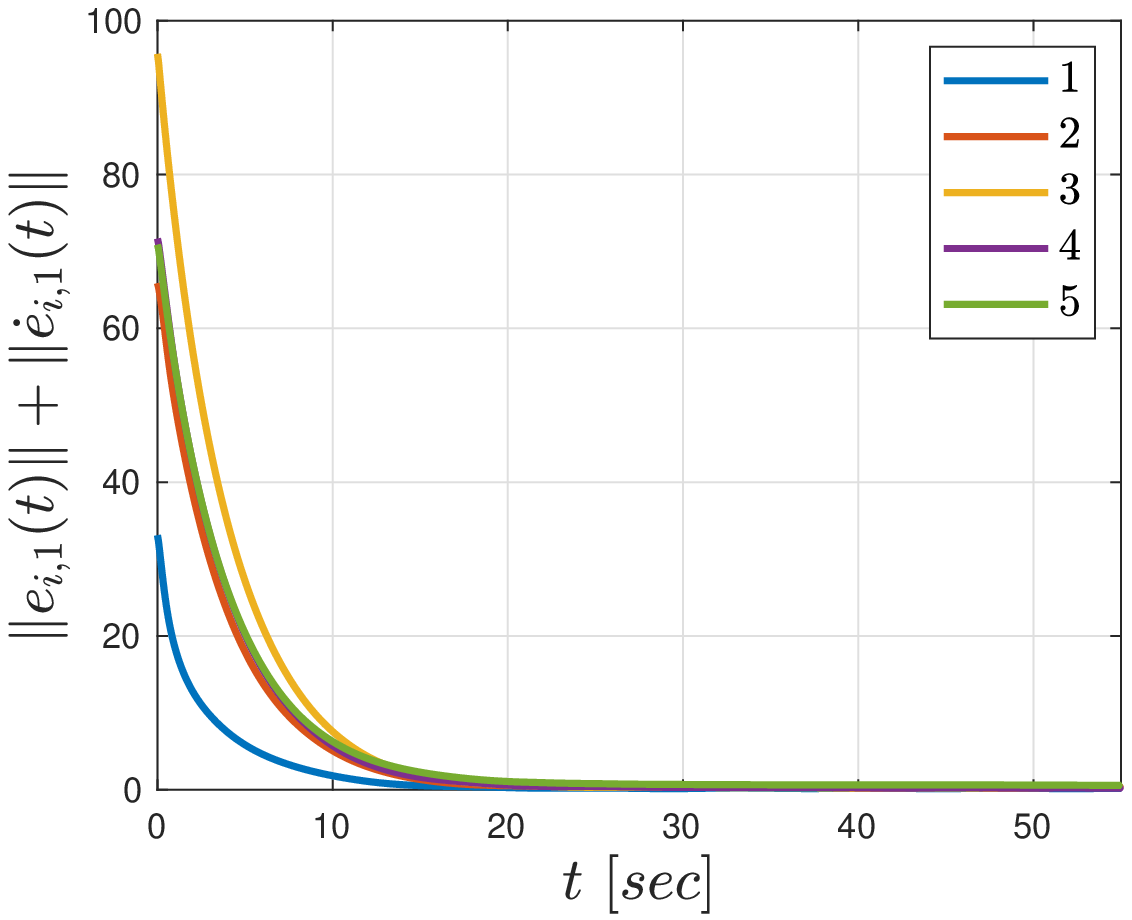}
	\includegraphics[width=.4\textwidth]{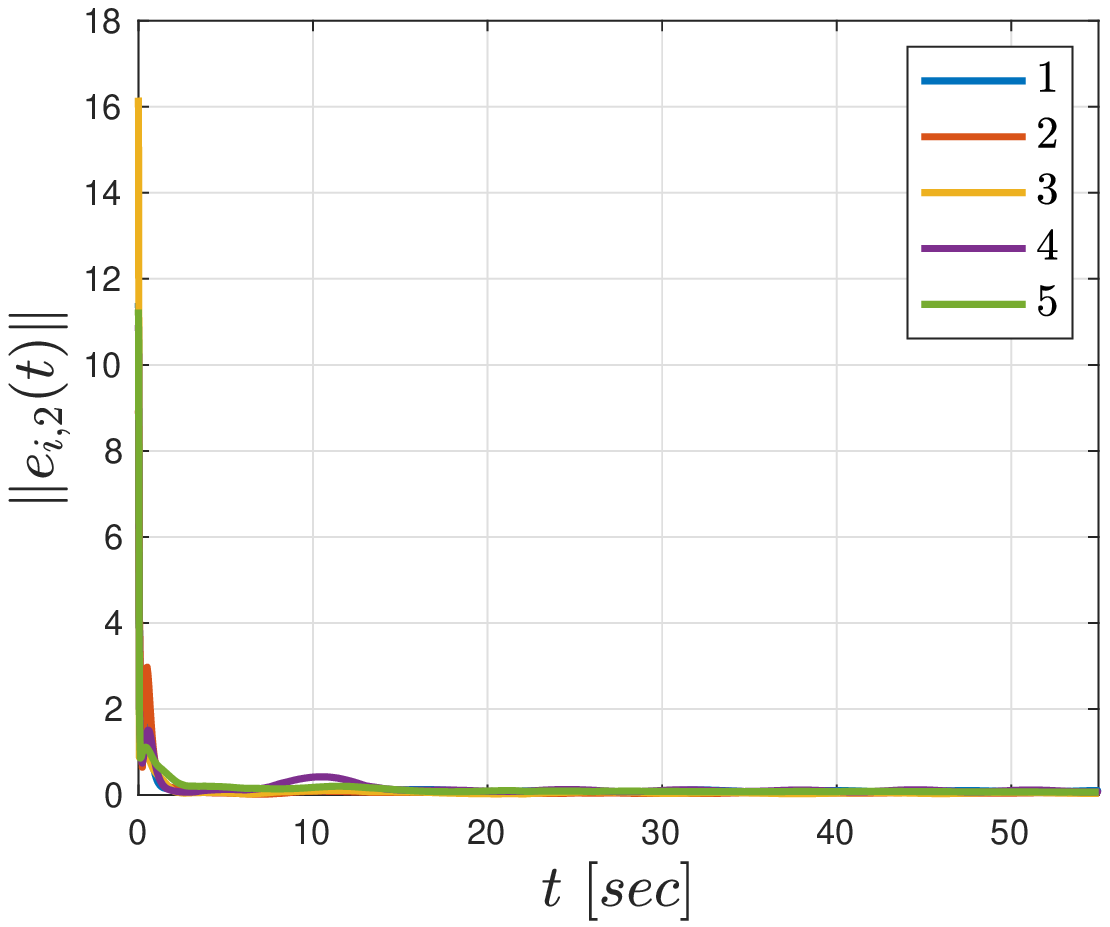}
	\caption{Evolution of the error signals $\|e_{i,1}(t)\|+\|\dot{e}_{i,1}(t)\|$, and $\|e_{i,2}(t)\|$,  for $i\in\{1,\dots,5\}$, and $t\in[0,55]$, in the first experiment.}
	\label{fig:AS errors}
\end{figure}

\begin{figure}
	\centering
	\includegraphics[width=.475\textwidth]{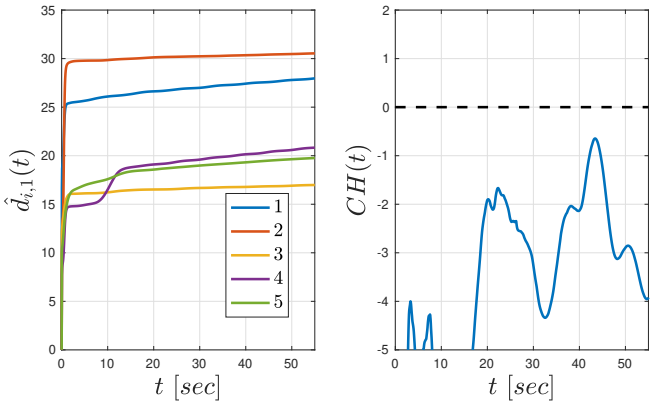}
	\caption{{Left: The evolution of the adaptation signals $\hat{d}_{i,1}(t)$ for $i\in\{1,\dots,5\}$, in the first experiment. Right: The evolution of $CH(t)$ in the first experiment.}}
	\label{fig:AS adaptation and CHECK}
\end{figure}

\begin{figure}
	\centering
	\includegraphics[width=.5\textwidth]{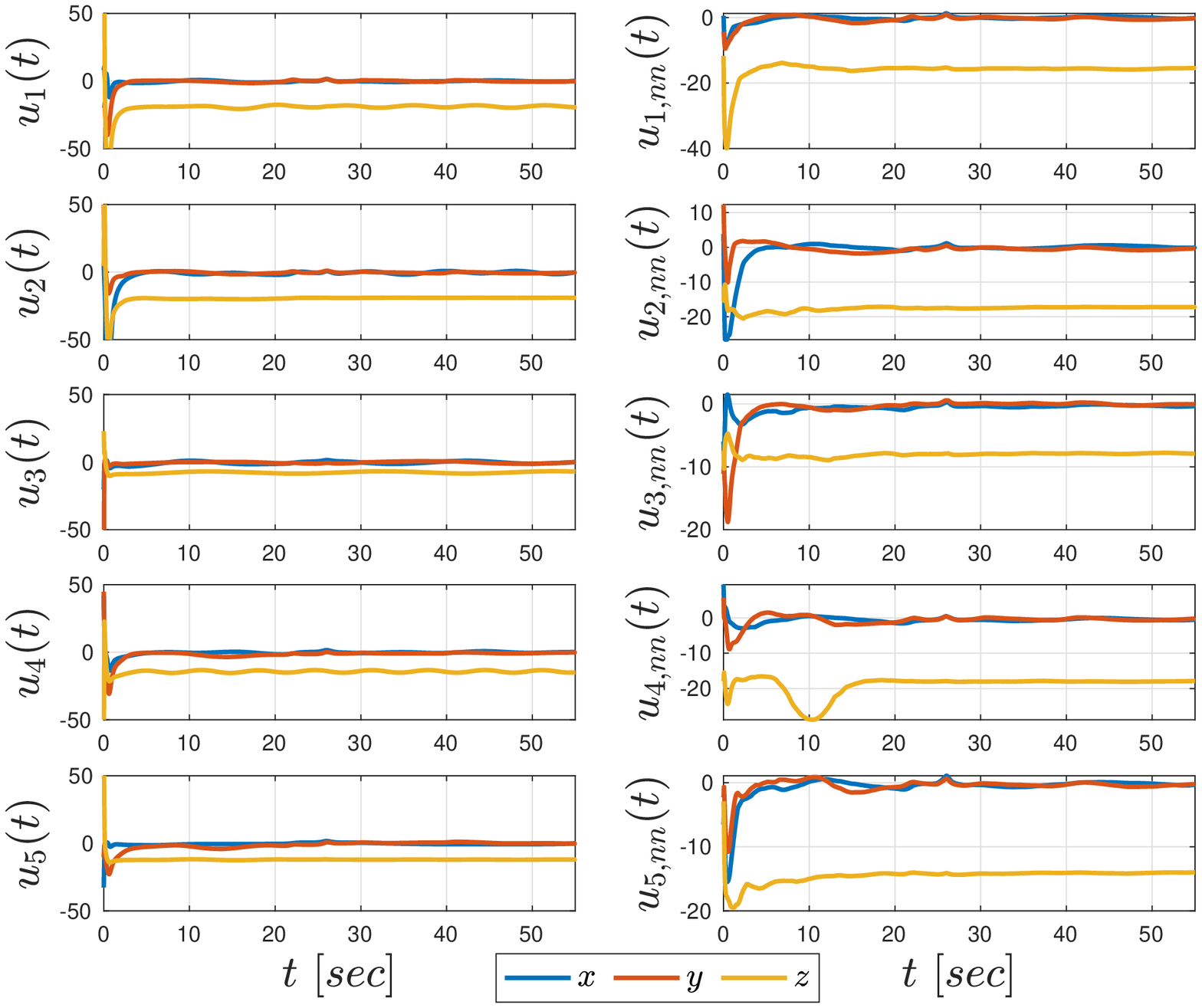}
	\caption{The evolution of the control inputs $u_{i}(t)$ and the neural-network controllers $u_{i,nn}(t)$, for $i\in\{1,\dots,5\}$, in the first experiment.}
	\label{fig:AS u}
\end{figure}

\begin{remark}
	{The control design is inspired by adaptive control methodologies \cite{krstic1995nonlinear}, where the time-varying coefficients $\hat{d}_{i,1}$ adapt, in coordination with the neural-network controllers, to the unknown dynamics in order to ensure closed-loop stability. In particular, by inspecting the proof of Theorem 1, it can be concluded that $\hat{d}_{i,1}$ aims to counteract the term $\frac{\|H^{-1}K_1\|}{\lambda_{\min}(g_i)} + \frac{\kappa}{\lambda_{\min}(g_i)}$, $i\in\mathcal{N}$. Intuitively, $\hat{d}_{i,1}$ increases according to \eqref{eq:adaptation law 1} until it dominates the aforementioned term, leading to convergence of $e_{i,2}$ to zero, for all $i\in\mathcal{N}$. }

{Note further that agent $i$'s control policy \eqref{eq:control and adapt} does not use any information on its own or its neighbors' dynamic terms $f_i(\cdot)$, $g_i(\cdot)$, or the constants $r$, $\kappa$ of \eqref{eq:nn cond}. 
Additionally, each agent uses only relative feedback from its neighbors, as can be verified by \eqref{eq:e error}, \eqref{eq:e_2} and \eqref{eq:control and adapt}. }
\end{remark}

 The following theorem, whose proof is given in the \hyperref[subsec:proof Th1]{appendix}, guarantees the accomplishment of the formation task.
\begin{theorem} \label{th:main}
	Let a multi-agent system evolve subject to the dynamics \eqref{eq:dynamics} under an undirected communication graph $\bar{\mathcal{G}}$.	 Under Assumptions {\ref{ass:g pd}-\ref{ass:nn}}, there exists a set $\bar{\Omega}_{\hat{x}} \subset \mathbb{R}^{N(2n+1)}$ such that, if $\big(e(0),\hat{d}_1(0)\big) \in \bar{\Omega}_{\hat{x}}$,  
the  distributed control mechanism guarantees $\lim_{t\to\infty}(e_{i,1},e_{i,2}) = 0$, for all $i\in\mathcal{N}$, as well as the boundedness of all closed-loop signals. 	 
\end{theorem}

Contrary to the works in the related literature (e.g., \cite{bechlioulis2016decentralized,verginis2019robust}) we do not impose reciprocal terms in the control input that grow unbounded in order to guarantee closed-loop stability. The resulting controller is essentially a simple linear feedback on $e_1$, $e_2$ with time-varying adaptive control gains, accompanied by the neural network output that ensures 
condition \eqref{eq:nn cond}. 

\begin{figure*}
	\centering
	\includegraphics[width=.9\textwidth]{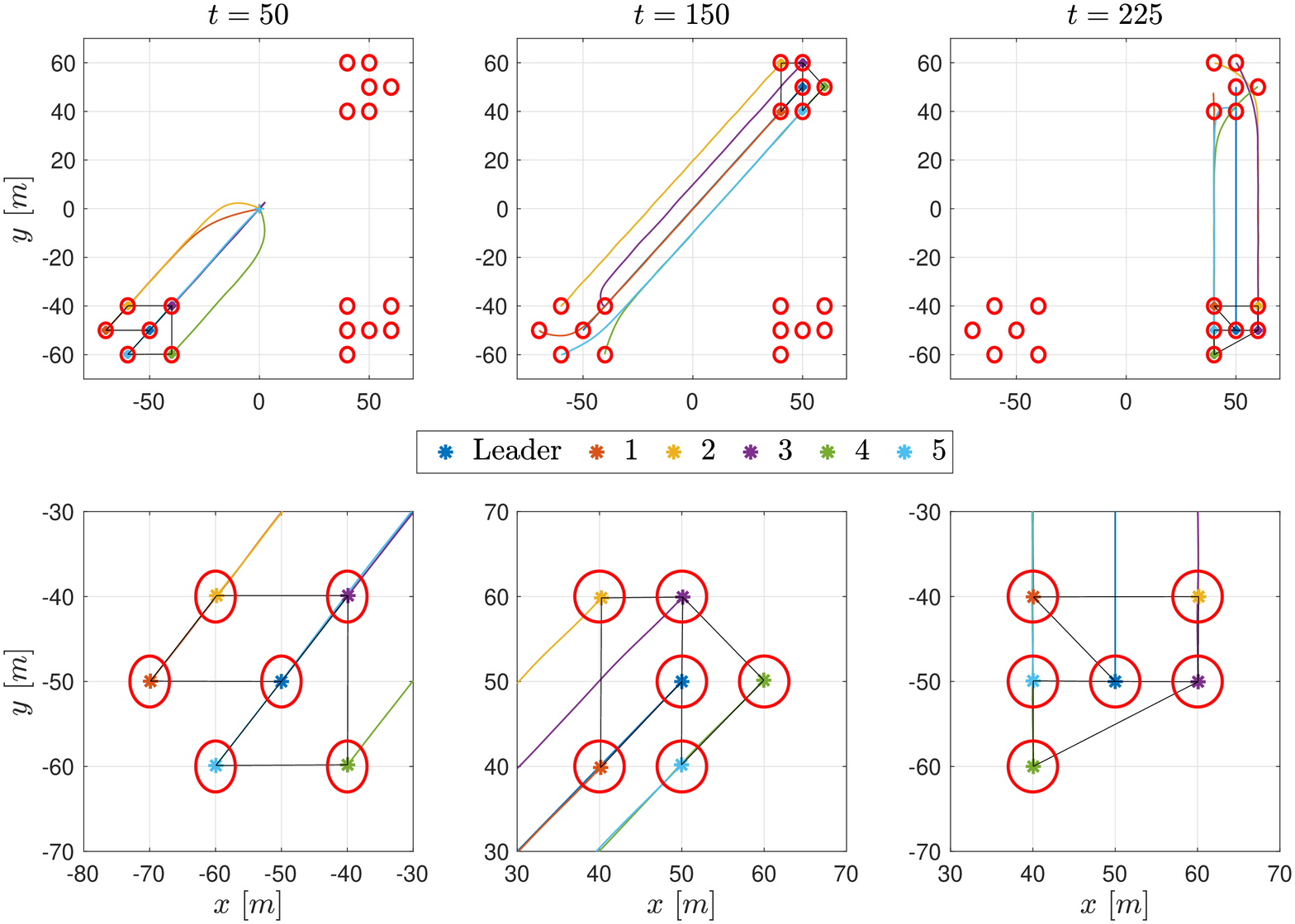}
	\caption{Snapshots of the second experiment (top) and their zoomed-in versions (bottom) in the $x$-$y$ plane. The agents converge to the desired formation around the leader at $t=50$, $t=150$, and $t=225$, which implies the visit of the regions of interest in the three areas. The black lines represent the communication edge set $\bar{\mathcal{E}}$ of the agents.
		The initial positions of the agents are depicted with $``+"$ in the top-left plot.}
	\label{fig:SURV xy plot}
\end{figure*}
\begin{figure}
	\centering
	\includegraphics[width=.35\textwidth]{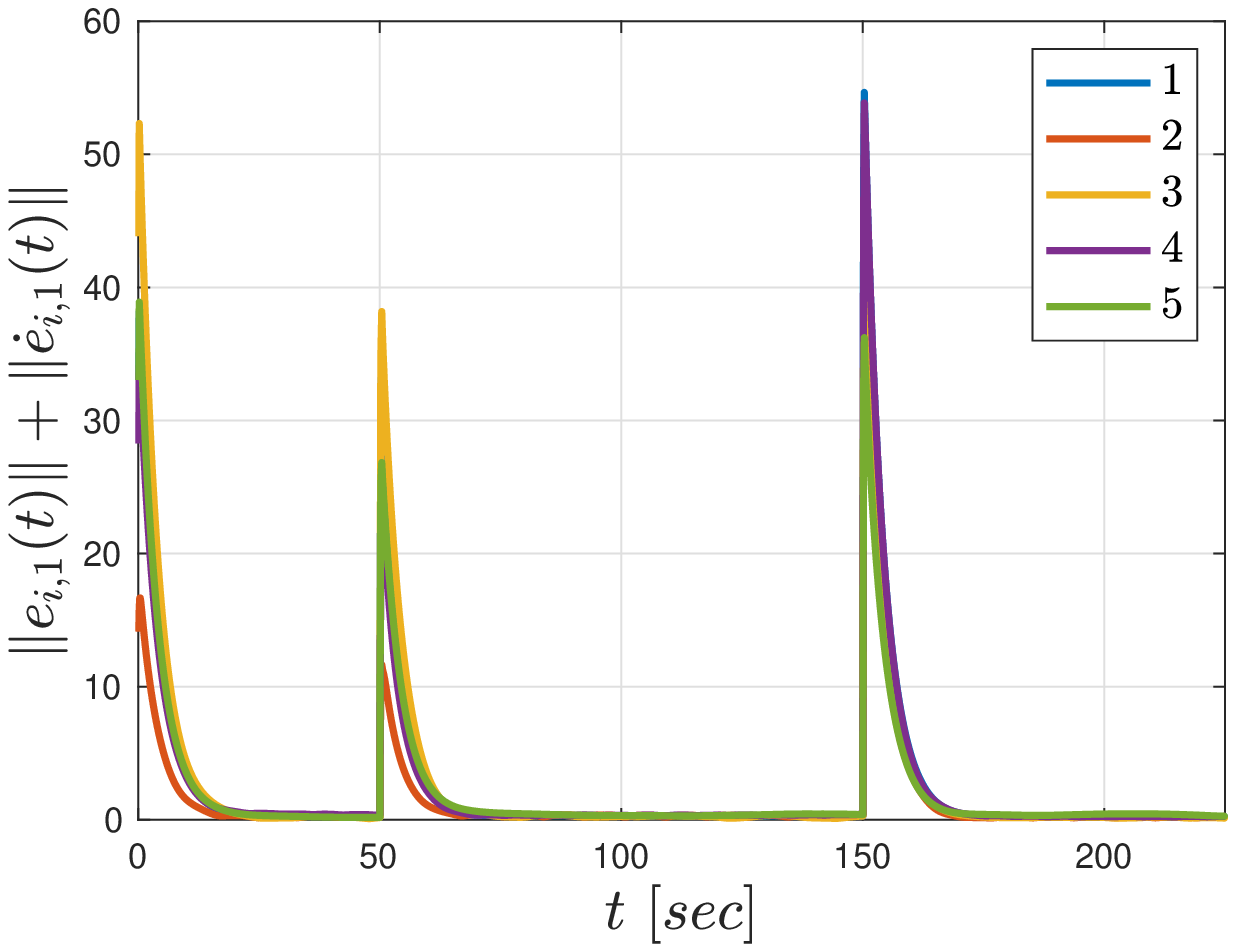}
	\includegraphics[width=.35\textwidth]{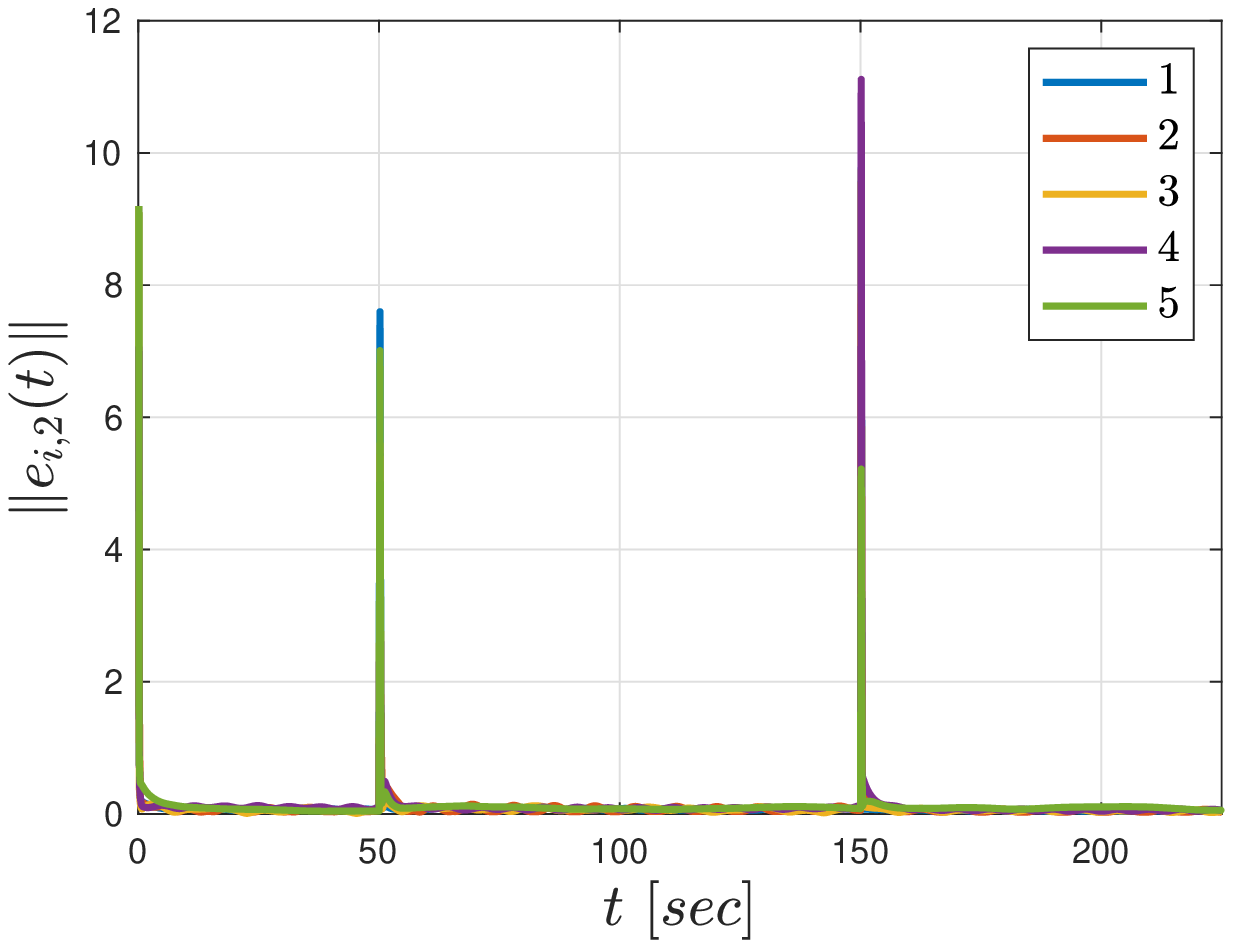}
	\caption{Evolution of the error signals $\|e_{i,1}(t)\|+\|\dot{e}_{i,1}(t)\|$ and $\|e_{i,2}(t)\|$, for $i\in\{1,\dots,5\}$, and $t\in[0,225]$, in the second experiment.}
	\label{fig:SURV errors}
\end{figure}

\begin{figure}
	\centering
	\includegraphics[width=.45\textwidth]{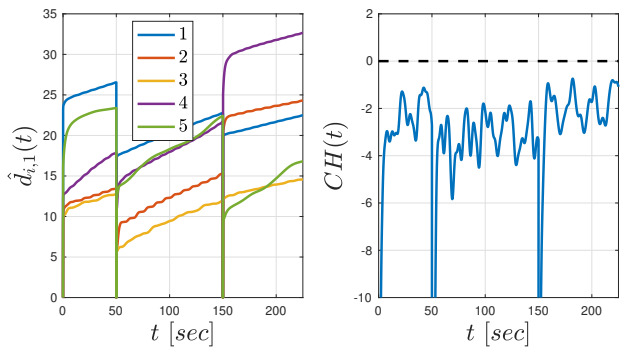}
	\caption{{Left: The evolution of the adaptation signals $\hat{d}_{i,1}(t)$ for $i\in\{1,\dots,5\}$, in the second experiment. Right: The evolution of $CH(t)$ in the second experiment.}}
	\label{fig:SURV adaptation}
\end{figure}

\begin{figure}
	\centering
	\includegraphics[width=.525\textwidth]{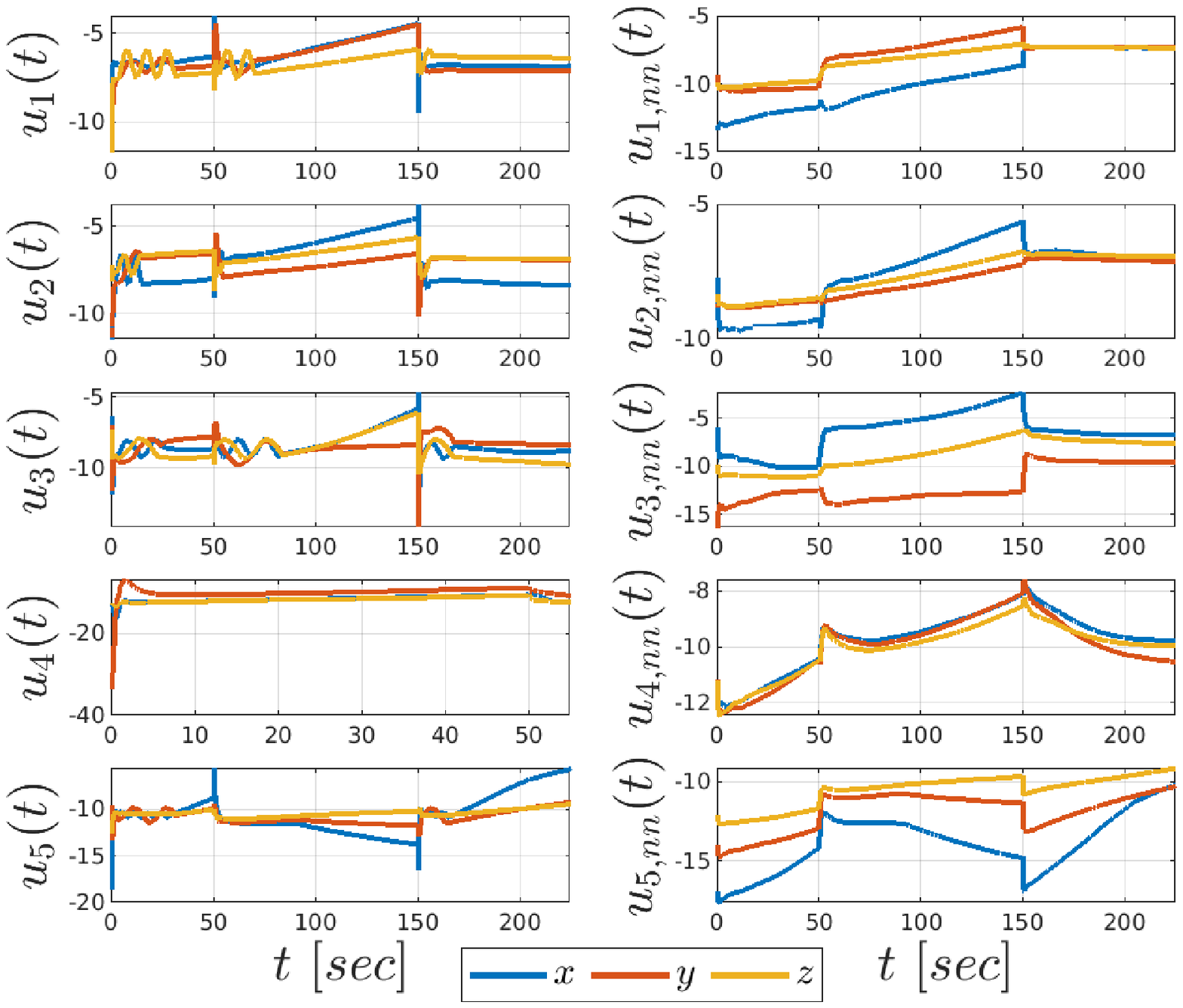}
	\caption{The evolution of the control inputs $u_{i}(t)$ and the neural-network controllers $u_{i,nn}(t)$, for $i\in\{1,\dots,5\}$, in the second numerical experiment.}
	\label{fig:SURV u}
\end{figure}

\section{Numerical Experiments} \label{sec:exps}
We consider $N=5$ follower aerial vehicles in $\mathbb{R}^3$ with dynamics of the form \eqref{eq:dynamics}, with 
\begin{align*}
	f_i(x_i,t) &= \frac{1}{m_i}(\bar{g}_r + d_{i,1}(t) +  d_{i,2}(x_i))  \\
	g_i(x_i,t) &= { \frac{ \|x_i\| + 0.5\sin(0.1t) + 1  }{m_i}}
\end{align*}
where $\bar{g}_r =[0,0,9.81]^\top$ is the gravity vector and $m_i \in \mathbb{R}$ is the mass of agent $i\in\mathcal{N}$. Furthermore, we choose $d_{i,1}(t)$, $d_{i,2}(x_i)$ as 
\begin{align*}
	d_{i,1}(t) &= 
	\begin{bmatrix}
		A_{i,1} \sin(\eta_{i,1}t + \phi_{i,1})  \\
		A_{i,2} \sin(\eta_{i,2}t + \phi_{i,2}) \\
		A_{i,3} \sin(\eta_{i,3}t + \phi_{i,3})  	
	\end{bmatrix} \\
	d_{i,2}(x_i) &= F_i y_i	
\end{align*}
with $y_i = [x_{i,2_1}^2, x_{i,2_2}^2, x_{i,2_3}^2, x_{i,2_1}x_{i,2_2}, x_{i,2_1}x_{i,2_3},x_{i,2_2},x_{i,2_3}]$, and we further use the notation $x_{i,2} = [x_{i,2_1},x_{i,2_2},x_{i,2_3}]^\top$ for all $i\in\mathcal{N}$. The terms $m_i$, $A_{i,\ell}$, $\eta_{i,\ell}$, $\phi_{i,\ell}$ are constants that take values in $(0,1)$; similarly, $F_i\in\mathbb{R}^{3 \times 6}$ is a constant matrix whose elements take values in $(0,1)$.
We evaluate the proposed algorithm in three test cases. In all of these cases, we choose the control gains of \eqref{eq:control and adapt} as $k_{i,1} = 0.1$, $k_{i,2} = \mu_{i,1} = 0.5$.

The first case consists of the stabilization of the followers around the leader, which is assigned with the tracking of a reference time-varying trajectory profile $x_0(t)$. We consider a communication graph modeled by the edge set $\bar{\mathcal{E}}$ $=$ $\{$ $(1,2)$, $(2,3)$, $(3,4)$, $(4,5)$, $(1,0)$, $(3,0)$, $(5,0)$ $\}$, i.e., agents $1$, $3$, and $5$ have access to the information of the leader.
The stabilization is dictated by the formation constants $c_{1,2} = -c_{2,1} = [1,1,0]^\top$,
$c_{2,3} = -c_{3,2}= [1,-1,0]^\top$,
$c_{3,4} = -c_{4,3}= [0,-2,0]^\top$,
$c_{4,5} = -c_{5,4}= [-2,0,0]^\top$,
$c_{1,0} = [1,-1,0]^\top$,
$c_{3,0} = [-1,-1,0]^\top$,
$c_{5,0} = [1,1,0]^\top$. 
 The aforementioned parameters, along with the agents' initial conditions, specify the first task's formation instance  $\mathcal{F} \coloneqq (x_0,f,g,c,\bar{\mathcal{G}},x(0))$.
We generate data from $100$ trajectories that correspond to different ${f}$, ${g}$, ${x}(0)$ than in $\mathcal{F}$, but with the same leader profile $x_0$ and inter-agent formation offsets ${c}$ and communication graph $\bar{\mathcal{G}}$. The differences in ${f}$ and ${g}$ are created by assigning random values, in $(0,1)$, to the constants $m_i$, $A_{i,\ell}$, $\eta_{i,\ell}$, $\phi_{i,\ell}$, and $F_i$, for all $i\in\mathcal{N}$. We further assign the initial conditions for each agent as $x_{i,1}(0) = x_{0,1}(0) + \textup{rand}(-4,4)[1,1,1]^\top$, and $x_{i,2}(0) = \textup{rand}(-2,2)[1,1,1]^\top$, $i\in\mathcal{N}$; we set the leader agent's initial condition as $x_{0,1}(0) = [5,2,10]^\top$, $x_{0,2}(0) = [0.0039,-09836,0]^\top$ for all trajectories. We use the generated data to
train $5$ neural networks, one for each agent. More details regarding the training can be found at the end of the section.
 We test the control policy \eqref{eq:control and adapt} using the task's formation instance $\mathcal{F}$. 
The results are depicted in Figs. \ref{fig:AS xy plot}-\ref{fig:AS u}; Fig. \ref{fig:AS xy plot} depicts snapshots of the multi-agent formation in the $x$-$y$ plane and Fig. \ref{fig:AS errors} shows the evolution of the error signals $\|e_{i,1}(t)\| + \|\dot{e}_{i,1}(t)\|$ and $\|e_{i,2}(t)\|$ for $i\in\{1,\dots,5\}$. Fig. \ref{fig:AS adaptation and CHECK} shows the evolution of the adaptation variables $\hat{d}_{i,1}(t)$, $i\in\mathcal{N}$, and the signal {$CH(t) = e_2(t)^\top  (f(x(t),t) + g(x(t),t)u_{{nn}}(x(t)) - \ddot{\bar{x}}_{0,1}(t)) - 100 \|e_2\|$, which is always negative, verifying thus that Assumption \ref{ass:nn} holds for $\kappa =100$.} Finally, Fig. \ref{fig:AS u} depicts the evolution of the control inputs $u_{i}(t)$, $u_{i,nn}(t)$, $i\in\{1,\dots,5\}$. 
One concludes that the multi-agent system converges successfully to the pre-specified formation, whose x-y shape is depicted in the bottom-right plot of Fig. \ref{fig:AS xy plot}. 

The second case comprises a surveillance task, where the agents need to periodically surveil three  areas in the environment. We choose the same communication graph as in the first case. Each area consists of 6 spherical regions of interest; the regions of interest of the first area are centered at $[-50,-50,-10]^\top$, $[-70,-50,10]^\top$, $[-60,-40,10]^\top$, $[-40,-40,10]^\top$, $[-40,-60,10]^\top$,  $[-60,-60,10]$; the regions of interest of the second area are centered at $[50,50,10]^\top$, $[40,40,10]^\top$, $[40,60,10]^\top$, $[50,60,10]^\top$, $[60,50,10]^\top$ $[50,40,10]^\top$; and the regions of interest of the third area are centered at $[50,-50,10]^\top$, $[40,-40,10]^\top$, $[60,-40,10]^\top$, $[60,-50,10]^\top$, $[40,-60,10]^\top$, $[40,-50,10]^\top$. 
The leader agent navigates sequentially to one of the regions in the areas, and by setting the constants $c_{ij}$, $(i,j)\in\bar{\mathcal{E}}$, according to the geometry of the regions, the followers aim to visit the remaining five regions in each area. More specifically, we set the formation constants as $c_{1,2} = -c_{2,1} = [10,10,0]^\top$,
$c_{2,3} = -c_{3,2}= [20,0,0]^\top$,
$c_{3,4} = -c_{4,3}= [0,-20,0]^\top$,
$c_{4,5} = -c_{5,4}= [-20,0,0]^\top$,
$c_{1,0} = [20,0,0]^\top$,
$c_{3,0}  = [-10,-10,0]^\top$,
$c_{5,0} = [10,10,0]^\top$ for the first area, 
$c_{1,2} = -c_{2,1} = [0,20,0]^\top$,
$c_{2,3} = -c_{3,2}= [10,0,0]^\top$,
$c_{3,4} = -c_{4,3}= [10,-10,0]^\top$,
$c_{4,5} = -c_{5,4}= [-10,-10,0]^\top$,
$c_{1,0} = [10,10,0]^\top$,
$c_{3,0}  = [0,-10,0]^\top$,
$c_{5,0} = [10,10,0]^\top$ for the second area,
and 
$c_{1,2} = -c_{2,1} = [20,0,0]^\top$,
$c_{2,3} = -c_{3,2}= [0,-10,0]^\top$,
$c_{3,4} = -c_{4,3}= [-20,-10,0]^\top$,
$c_{4,5} = -c_{5,4}= [0,10,0]^\top$,
$c_{1,0} = [10,-10,0]^\top$,
$c_{3,0}  = [-10,0,0]^\top$,
$c_{5,0} = [10,0,0]^\top$ for the third area.

Similarly to the first case, we generate data from $100$ trajectories that correspond to different ${f}$, ${g}$, ${x}(0)$ than in the task's formation instance $\mathcal{F}$; the differences in ${f}$, ${g}$ are created by assigning random values, in $(0,1)$, to the constants $m_i$, $A_{i,\ell}$, $\eta_{i,\ell}$, $\phi_{i,\ell}$, and $F_i$, for all $i\in\mathcal{N}$. The initial conditions of the agents are set as $x_{i,1}(0) = x_{0,1}(0) + \textup{rand}(-10,10)[1,1,1]^\top$, and $x_{i,2}(0) = \textup{rand}(-2,2)[1,1,1]^\top$, $i\in\mathcal{N}$, and of the leader agent as $x_{0,1}= [0,0,10]^\top$, $x_{0,2} = [0,0,0]^\top$. We use the data to train $5$ neural networks, one for each agent. We test the control policy \eqref{eq:control and adapt} on $\mathcal{F}$, giving the results depicted in Figs. \ref{fig:SURV xy plot}-\ref{fig:SURV u}; Fig. \ref{fig:SURV xy plot} depicts snapshots of the agents' visit to the three areas (at $t=50$, $t=150$, and $t=225$ seconds, respectively), and Fig. \ref{fig:SURV errors} depicts the evolution of the signals $\|e_{i,1}(t)\| + \|\dot{e}_{i,1}(t)\|$ and $\|\dot{e}_{i,2}(t)\|$, for all agents $i\in\{1,\dots,5\}$. Fig. \ref{fig:AS adaptation and CHECK} shows the evolution of the adaptation variables $\hat{d}_{i,1}(t)$, $i\in\mathcal{N}$, and the signal {$CH(t) = e_2(t)^\top  (f(x(t),t) + g(x(t),t)u_{nn}(x(t)) - \ddot{\bar{x}}_{0,1}(t)) - 100 \|e_2\|$, which is always negative, verifying thus that Assumption \ref{ass:nn} holds for $\kappa =100$.} Finally, Fig. \ref{fig:AS u} depicts the evolution of the control inputs $u_{i}(t)$, $u_{i,nn}(t)$, $i\in\{1,\dots,5\}$. 
As illustrated in the figures, the agents converge successfully to the three pre-specified formations, visiting the regions of interest in the three areas.

The first two cases considered training data that correspond to the exact formation task, defined by the leader profile $x_0$ and the constants $c_{ij}$, and communication graph $\bar{\mathcal{G}}$. In the third case, we generate $120$ different formation instances $\mathcal{F}^k \coloneqq (x^k_0,{f}^k,{g}^k,{c}^k,\bar{\mathcal{G}}^k,{x}^k(0))$, $k\in\{1,\dots,120\}$, i.e., different trajectory profiles for the leader, different terms ${f}^k$ and ${g}^k$ for the agents, different communication graphs $\bar{\mathcal{G}}$, different formation constants $c_{ij}$, for $(i,j)\in\bar{\mathcal{E}}$, and different initial conditions for the agents. In every instance $k$, we set the parameters in ${f}^k$, and ${g}^k$ as in the previous two cases, we set randomly the communication graph $\bar{\mathcal{G}}^k$ such that it satisfies Assumption \ref{ass:graph connect}, we set random offsets $c_{ij}$ in the interval $(-5,5)\bar{1}_3$, for $(i,j)\in\bar{\mathcal{E}}$, and the initial conditions of the agents as $x_{i,1}(0) = \textup{rand}(-10,10)\bar{1}_3$, $x_{i,2}(0) = \textup{rand}(-2.5,2.5)\bar{1}_3$, for all $i\in\{1,\dots,5\}$. Finally, we set the leader trajectory $x_0$ for each instance $k\in\{1,\dots,120\}$ as follows: we create four points in $\mathbb{R}^3$ randomly in $(-10,10)$ in the $x$- and $y$- directions, and in $(1,20)$ in the $z$ direction. We then create a random sequence of these points, and set the leader trajectory as a smooth path that visits them according to that sequence, with a duration of $40$ seconds.

We separate the $120$ instances into $100$ training and $20$ test instances. We train next $5$ neural networks, one for each agent, using data from system runs that correspond to the $100$ first training instances $\mathcal{F}^k$, $k\in\{1,\dots,100\}$. We test the control policy on the $20$ first training instances $\mathcal{F}^k$, $k\in\{1,\dots,20\}$, as well as on the $20$ test instances that were not used in the training, i.e., $\mathcal{F}^k$, $k\in\{101,\dots,120\}$. In addition, we compare the performance of the proposed control algorithm with a \textit{no-neural-network} (no-NN) control policy, i.e., a policy that does not employ the neural network, (term  $u_{i,nn}$ in \eqref{eq:control law}) and with a non-adaptive control policy $u_{i} = u_{i,nn} - k_{i,2}e_{i,2}$, i.e., without the adaptation terms $\hat{d}_{i,1}$. The comparison results are given in Fig. \ref{fig:comparison errors}, which depicts the mean and standard deviation of the signal $\|{e}_1(t)\| + \|\dot{{e}}_1(t)\|$ for the 20 of the training instances (top), and for the 20 test instances (bottom). It can be verified that, in both cases, the proposed control algorithm outperforms the other two policies, which, in many of the instances, resulted in unstable closed-loop systems. 

We now provide more details regarding the collection of data and the training of the neural networks for the aforementioned experiments. For the execution of the trajectories that are used in the training of the neural networks, we use the control policies
\begin{align*}
	u_i = g_i(x_i,t)^{-1}(u_0(t) - e_{i,2} - f_i(x_i,t)),
\end{align*}
for all $i\in\mathcal{N}$. The data for the training of the neural networks consist of 100 system trajectories, sampled at 500 points, making a total of 50000 points. The neural networks we use consist of 4 fully connected layers of 512 neurons; each layer is followed by a batch-normalization module and a ReLU activation function. For the training, we use the Adam optimizer, the mean-square-error loss function, and learning rate of $10^{-3}$. Finally, we use a batch size of 256, and we train the neural networks until an average (per batch) loss of the order of $10^{-4}$ is achieved.

\begin{figure}
	\includegraphics[width=.5\textwidth]{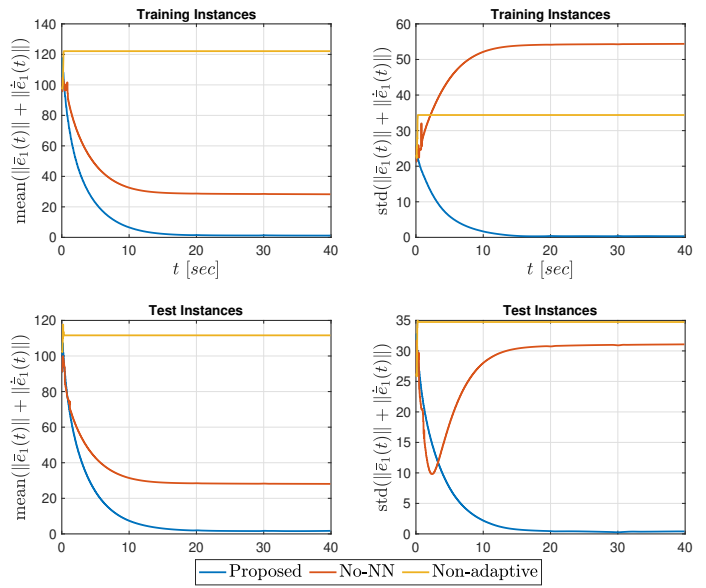}
	\caption{Evolution of the mean (left) and standard deviation (right) of the signals $\|\bar{e}_1(t)\| + \|\dot{\bar{e}}_1(t)\|$ for the 20 training instances (top) and the 20 test instances (bottom). }
	\label{fig:comparison errors}
\end{figure}

%

\section{Conclusion and Future Work} \label{sec:concl}

We develop a  learning-based control algorithm for the  formation control of networked multi-agent systems with unknown nonlinear
dynamics. The algorithm integrates
distributed neural-network-based learning and adaptive control. We provide formal guarantees and perform
extensive numerical experiments. Future efforts will focus on relaxing the considered assumptions and extending the proposed methodology
to account for directed and time-varying communication graphs as well as underactuated systems.

\bibliographystyle{IEEEtran}
\bibliography{bibliography}

\appendix
\section{Appendix} \label{sec:app}

We provide here the proof of Theorem \ref{th:main}.

\begin{proof}[Proof of Theorem \ref{th:main}] \label{subsec:proof Th1}
	
	{ 
	Let the continuously differentiable function
	\begin{align} \label{eq:V_1}
		V_1  \coloneqq \frac{1}{2\underline{g}}{e}_1^\top K_1^2 H^{-1}  e_1 + \frac{1}{2\underline{g}}{e}_2^\top H^{-1}{e}_2. 
	\end{align}
	By differentiating $V_1$ and using \eqref{eq:e dynamics}, one obtains	
	\begin{align*}
		\dot{V}_1 =& -\frac{1}{\underline{g}}e_1^\top K_1^2H^{-1}K_1 e_1 + \frac{1}{\underline{g}}e_1^\top K_1^2 H^{-1} e_2 \\& + \frac{1}{\underline{g}}e_2^\top\left( f(x,t) + g(x,t)u - \ddot{\bar{x}}_{0,1} \right) \\&
		- \frac{1}{\underline{g}}e_2^\top H^{-1} K_1^2 e_1 + \frac{1}{\underline{g}}e_2^\top H^{-1} K_1 e_2 
	\end{align*}
	and by further using \eqref{eq:control law},
	\begin{align*}
	 \dot{V}_1 \leq& -\frac{1}{\underline{g}}{e}_1^\top {K}_1^2 H^{-1} K_1{e}_1 + \frac{\|H^{-1}K_1\|}{\underline{g}}\sum_{i\in\mathcal{N}}\|e_{i,2}\|^2 \\
	&  - \frac{1}{\underline{g}}\sum_{i\in\mathcal{N}}e_{i,2}^\top g_i(x_i,t)(k_{i,2} + \hat{d}_{i,1})e_{i,2} \\
	& + \frac{1}{\underline{g}}e_2^\top\left( f(x,t) + g(x,t)u_{nn}(x) - \ddot{\bar{x}}_{0,1} \right)
	\end{align*}
	By using the positive definiteness of $g_i(x_i,t)$, the fact that  $\underline{g} = \min_{i\in\mathcal{N}}\{ \lambda_{\min}(g_i) \}$, and the fact that $\hat{d}_{i,1}(t)$ is positive, $i\in\mathcal{N}$, we obtain
		\begin{align*} 
	\dot{V}_1 \leq& -\frac{1}{\underline{g}}{e}_1^\top {K}_1^2 H^{-1} K_1{e}_1 + \frac{\|H^{-1}K_1\|}{\underline{g}}\sum_{i\in\mathcal{N}}\|e_{i,2}\|^2 \notag \\
	&  - \sum_{i\in\mathcal{N}}(k_{i,2} + \hat{d}_{i,1})\|e_{i,2}\|^2 \notag \\
	& + \frac{1}{\underline{g}}e_2^\top\left( f(x,t) + g(x,t)u_{nn}(x) - \ddot{\bar{x}}_{0,1} \right) 
	\end{align*}
and in view of Assumption \ref{ass:nn}, for $\|e\| \leq r$, 
	\begin{align} \label{eq:W_1_tilde}
	\dot{V}_1 \leq& -\frac{1}{\underline{g}}{e}_1^\top {K}_1^2 H^{-1} K_1{e}_1 \notag \\
	&\hspace{-10mm}  - \sum_{i\in\mathcal{N}}\left(k_{i,2} + \hat{d}_{i,1} - \frac{\|H^{-1}K_1\|}{\underline{g}} - \frac{\kappa}{\underline{g}}\right)\|e_{i,2}\|^2  
	\end{align}
	}
	
	{
By further defining $d_1 \coloneqq \frac{\|H^{-1}K_1\|}{\underline{g}} + \frac{\kappa}{\underline{g}}$, \eqref{eq:W_1_tilde} becomes 
\begin{align} \label{eq:W_1_tilde 2}
	\dot{V}_1 \leq& -\frac{1}{\underline{g}}{e}_1^\top {K}_1^2 H^{-1} K_1{e}_1  - \sum_{i\in\mathcal{N}}(k_{i,2} + \hat{d}_{i,1} - d_1)\|e_{i,2}\|^2
	\end{align}
	}

	{In view of the aforementioned expression, the individual adaptation variables $\hat{d}_{i,1}$ aim to dominate the term $d_1$. Therefore, we define the adaptation errors $\widetilde{d}_1 \coloneqq [\widetilde{d}_{1,1},\dots,\widetilde{d}_{N,1}]^\top$ $\coloneqq$ $\hat{d}_1 - \bar{d}_1$ $\coloneqq$ $[\hat{d}_{1,1}-d_1,\dots,\hat{d}_{N,1}-d_1]^\top$, 
and the overall state $\widetilde{x}$ $\coloneqq$ $[{e}_1^\top,{e}_2^\top,\widetilde{d}^\top_1]^\top$ $\in$ $\mathbb{R}^{N(2n+1)}$. Let  the continuously differentiable function 
	\begin{align*}
		V_2(\widetilde{x})  \coloneqq V_1(\widetilde{x}) +  \frac{1}{2}\widetilde{d}_1^\top {M}_1^{-1}\widetilde{d}_1,
	\end{align*}
where ${M}_1 \coloneqq \textup{diag}\{\mu_{1,1},\dots,\mu_{N,1}\}$.	
	By differentiating $V_2$ and using \eqref{eq:W_1_tilde 2}, we obtain  
	\begin{align*}
		\dot{V}_2 \leq & -\frac{1}{\underline{g}}{e}_1^\top {K}_1^2 H^{-1} K_1{e}_1  - \sum_{i\in\mathcal{N}}(k_{i,2} + \hat{d}_{i,1} - d_1)\|e_{i,2}\|^2  \notag \\
	 &   +  \sum_{i\in\mathcal{N}} \frac{1}{\mu_{i,1}}\widetilde{d}_{i,1} \dot{\hat{d}}_{i,1} 
	\end{align*}
	and by substituting \eqref{eq:adaptation law 1},
	\begin{align*}
		\dot{V}_2  \leq& -\frac{1}{\underline{g}}{e}_1^\top {K}_1^2 H^{-1} K_1{e}_1  - \sum_{i\in\mathcal{N}} {k}_{i,2} \|{e}_{i,2}\|^2 \leq 0
	\end{align*} 
	Therefore, $V_2(t) \leq V_2(0)$, implying the boundedness of $e_1(t)$, $e_2(t)$, and $\widetilde{d}_1(t)$, for all $t\geq 0$. In view of \eqref{eq:control and adapt},  we also conclude the boundedness of $u(t)$ and $\dot{\hat{d}}_1(t)$, for all $t\geq 0$. By differentiating $\dot{V}_2$ and using  \eqref{eq:e dynamics} and \eqref{eq:control and adapt},  we further conclude the boundedness of $\ddot{V}_2(t)$, $t\geq 0$, which implies the uniform continuity of $V_2$. By employing Barbalat's Lemma (Theorem 8.4 of \cite{khalil1996noninear}), we conclude that $\lim_{t\to\infty}e_1(t) = \lim_{t\to\infty}e_2(t)  = 0$.}
	
{
In view of Assumptions \ref{ass:g pd} and \ref{ass:nn}, the aforementioned results hold under the conditions $x\in\Omega_x \coloneqq \Omega_1 \times \dots \times \Omega_N$ and $\|e\| \leq r$. Therefore, we need to establish that the 	proposed control algorithm and initial conditions do not force $e(t)$ and $x(t)$ to exit the sets $\{e \in \mathbb{R}^{2Nn}: \|e\| \leq r\}$ and $\Omega_x$, respectively,
at any point in time $t \geq 0$. 
Alternatively, we need to establish that, for $\widetilde{x}(0) \in \bar{\Omega}$, it holds that $x(t) \in \Omega_x$ and $\|e(t)\| \leq r$, for all $t\geq 0$.  
Let the set
	\begin{align*}
		\mathcal{M} \coloneqq & \{ \widetilde{x} \in \mathbb{R}^{N(2n+1)}: V_2(\widetilde{x}) \leq V_0 \},
	\end{align*}
where we choose $V_0$ as the largest constant for which $\mathcal{M} \subseteq \{ \widetilde{x} \in \mathbb{R}^{N(2n+1)}: {x} \in \Omega_x, \|e\| \leq r, \widetilde{d}_1 \leq V_2(\widetilde{x}(0))  \}$. Then, for all $\widetilde{x}(0) \in \bar{\Omega}$, where $\bar{\Omega}\subseteq \mathcal{M}$, it follows that $V_2$ is bounded from above by $V_2(\widetilde{x}(0))$, which implies that ${x}(t) \in \Omega_x$  and $\|e(t)\| \leq r$, for all $t\geq 0$. Since $\widetilde{x}=[e^\top, \widetilde{d}_1]^\top=[e^\top, \hat{d}_1-\bar{d}_1]^\top$ and $\bar{d}_1$ is constant, $\widetilde{x}(0) \in \bar{\Omega}$ implies $[e(0)^\top,\hat{d}_1(0)^\top]^\top \in \bar{\Omega}_{\hat{x}} \coloneqq \{[e^\top,\hat{d}_1^\top]^\top \in \mathbb{R}^{N(2n+1)} :\widetilde{x} \in \Omega_x \}$, leading to the conclusion of the proof.}

\end{proof}

\end{document}